\documentclass[pra,floatfix,amsmath,superscriptaddress,twocolumn]{revtex4-2}
\usepackage{amssymb}
\usepackage{graphicx}
\usepackage{graphics}
\usepackage{amsmath}
\usepackage{amsthm}
\usepackage{color}

\newcommand{\bea}{\begin{eqnarray}}
\newcommand{\eea}{\end{eqnarray}}

\def\bi{\begin{itemize}}
\def\ei{\end{itemize}}
\def\bc{\begin{center}}
\def\ec{\end{center}}

\def\C{\hbox{$\mit I$\kern-.7em$\mit C$}}
\def\R{\hbox{$\mit I$\kern-.6em$\mit R$}}

\newcommand{\one}{\mbox{$1 \hspace{-1.0mm}  {\bf l}$}}
\def\tr{\mathrm{tr}}

\DeclareMathOperator{\diam}{diam}
\DeclareMathOperator{\id}{id}

\newcommand{\dmin}{\delta_\mathrm{min}}
\newcommand{\dmax}{\delta_\mathrm{max}}

\newcommand{\ghz}{\mathrm{GHZ}_n}

\newtheorem{theorem}{Theorem}
\newtheorem{corollary}[theorem]{Corollary}
\newtheorem{lemma}[theorem]{Lemma}

\newtheorem{definition}[theorem]{Definition}

\begin{document}

\title{Distillation of multipartite entangled states for arbitrary subsets of parties in noisy quantum networks of increasing size}
\author{Aitor Balmaseda}\email{abalmase@math.uc3m.es}
\affiliation{Departamento de Matem\'aticas, Universidad Carlos III de
Madrid, E-28911, Legan\'es (Madrid), Spain}
\author{Julio I. de Vicente}\email{jdvicent@math.uc3m.es}
\affiliation{Departamento de Matem\'aticas, Universidad Carlos III de
Madrid, E-28911, Legan\'es (Madrid), Spain}
\affiliation{Instituto de Ciencias Matem\'aticas (ICMAT), E-28049 Madrid, Spain}

\begin{abstract}
Quantum network states are multipartite states built from distributing pairwise entanglement among parties and underpin the paradigm of quantum networks for quantum information processing. In this work we introduce the problem of partial distillability in noisy quantum networks. This corresponds to the possibility of distilling by local operations and classical communication an arbitrary pure state for an arbitrary subset of parties starting from a single copy of a quantum network state of an increasing number of parties that share noisy (i.e.\ mixed) bipartite entangled states. Here, distillation means that the target state is obtained with fidelity as close to 1 as desired as the size of the network increases when the pairwise entangled links support a constant amount of noise. While we prove an obstruction to multipartite distillation protocols after teleportation with channels with constant noise, we show that partial distillability is indeed possible if certain well-established graph-theoretic parameters that measure the connectivity in the network grow fast enough with its size. We give necessary as well as sufficient conditions for partial distillability in terms of these parameters and we moreover provide explicit constructions of networks with partial distillability and a relatively slow connectivity growth.
\end{abstract}

\maketitle

\section{Introduction}
Multipartite entanglement is an essential resource for quantum information processing enabling tasks such as measurement-based quantum computation \cite{1wayqc}, quantum secret sharing \cite{secretsharing}, conference key agreement \cite{qcka} or quantum sensing \cite{gmesensing1,gmesensing2}. However, while in the bipartite case there is a maximally entangled state, which is therefore the most useful state in all entanglement-related applications, no such notion is possible in general for three or more parties \cite{mes}. Thus, for instance, the so-called GHZ state makes quantum secret sharing and conference key agreement possible but it is useless for measurement-based quantum computation \cite{noghz}, where, for example, the so-called cluster state is on the contrary a universal resource. 

Nevertheless, even if one has a well-defined task to accomplish and, hence, a particular state as target, the distribution of high-fidelity entangled states to many, possibly far-away, parties remains a formidable challenge. One of the problems one has to deal with is the fragility of the particular highly entangled pure states a specific application requires. The unavoidable introduction of noise renders states mixed, diminishing, and even destroying, their resource content. A thoroughly studied way to cope with this issue is given by entanglement distillation (see the review \cite{distillation} and references therein). These protocols act locally on several copies of a noisy state producing fewer copies of a purified state with the promise that the fidelity with the pure target state is as close to 1 as desired if the number of initial copies is sufficiently large and the noise is not too high. Unfortunately, it is not a priori clear how to distill to a particular target and these schemes need to be tailor-made for each case; however, many such protocols have been developed in the literature for different multipartite states of relevance \cite{distillation}. 

Another major obstacle is the distribution itself of the states as this requires a joint quantum channel that can send reliably quantum communication to many distant parties. Quantum networks are gaining considerable attention in this respect since they provide a relatively feasible way to achieve this \cite{qinternet}. The main idea here is, instead of directly aiming at creating an entangled state shared among all parties, one focuses initially at establishing entangled links only between pairs of parties, which is arguably easier (e.g.\ because there is an underlying topology given by distance constraints). Then, by combining distillation, teleportation and other techniques one can in principle use the established network to distribute any multipartite entangled state of choice by local operations and classical communication (LOCC). For this to work in a noisy environment, the bipartite entangled links need to be constantly refreshed and many recent works study specific protocols to distribute multipartite entanglement in networks using as figures of merit the number of perfect links consumed, waiting time, and final obtained fidelity together with potential sources of error (see e.g.\ \cite{networks1,networks2,networks3,networks4} and the review \cite{reviewn}).

Nevertheless, a static approach is also possible in which one views the quantum network as a fixed given resource state in which the entangled links can no longer be replenished. This perspective is more common from the abstract point of view of entanglement theory and underpins vastly explored phenomena such as entanglement percolation (see the review \cite{percolationreview} and references therein) and network-Bell-nonlocality (see the review \cite{bellnetworks} and references therein). Additionally, recent works study what LOCC transformations are possible from a single copy of a network state \cite{yamasakin,speen} and these states are sometimes used as convenient constructions to display a certain desired property such as maximal pairability \cite{pairable}, genuine multipartite nonlocality \cite{networkGMNL} or superactivation thereof \cite{AGME1}. While most of these works consider pure entangled links, it is of clear practical interest to investigate the resource nature of noisy network states (i.e.\ when the bipartite entangled connections become mixed). This is the approach taken in \cite{AGME1, AGME2}, where the ability of the network state to display asymptotically robust genuine multipartite entanglement (GME) has been shown to depend drastically on the properties of the graph that codifies the connectivity pattern of the network. In these models, referred to as isotropic quantum networks, the mixedness of the links is captured by a parameter $p\in[0,1]$ called visibility, where $p=1$ corresponds to perfect bipartite maximally entangled connections. The aforementioned result states that for some networks, for any given $p<1$ the corresponding network state is not GME for a sufficiently big size, while others are GME for any number of parties if the visibility is above a certain fixed threshold. This indicates that certain network patterns are severely limited in practical applications, where perfect visibility is an unjustified idealization.  

A relevant question in this context is whether starting from a noisy $N$-partite network state one can obtain by LOCC a particular $N$-partite pure entangled state. However, this is in general impossible because entanglement distillation can only yield a pure state from a mixed one in the asymptotic limit of infinitely many initial copies \cite{fang}. Nevertheless, this does not forbid purification for a fixed subset of $m$ parties from a single copy in the asymptotic limit $N\to\infty$ in which the network becomes arbitrarily large. This is the main subject of study of this work. That is, we consider whether it is possible starting from a big enough noisy network state to redistribute the mixed entangled links so as to concentrate the entanglement into a particular subset of parties and obtain a desired $m$-partite pure entangled state with fidelity arbitrarily close to one for $N$ sufficiently large. Moreover, we would like this to happen for an arbitrary subset of parties of arbitrary cardinality $m$ (i.e. for any given $m$ independent of $N$) and for an arbitrary pure $m$-partite target state with a fixed threshold (i.e.\ independent of the subset of parties, $m$ and the target state) in the visibility of the entangled links that constitute the network. In this way, if the entangling devices are sufficiently developed so as to guarantee a non-perfect but good enough visibility, then we can be sure that the task can be accomplished independently of the particular parties and the particular application (i.e.\ target state) we might want to aim for later. We refer to the property described above as partial distillability. This is in a certain sense analogous to the case of classical networks, where one implements some task involving a subset of nodes by exploiting all connections in the network letting the information flow collaboratively through all the nodes that constitute it.

Our main result is that partial distillability is in fact possible provided the graph sequence that encodes how the network grows is sufficiently connected in some precise sense. In more detail, certain $N$-partite network sequences have the property that if the visibility $p$ is above a certain universal threshold that only depends on the network configuration and the underlying dimension, then one can obtain by LOCC any $m$-partite state for any subset of $m$ parties for every fixed $m$ with fidelity going to 1 as $N\to\infty$. On the other hand, other network patterns cannot achieve this robustly, i.e.\ for any given visibility $p<1$ there exists $m\in\mathbb{N}$ such that distillation to $m$-partite states is impossible for at least one subset of $m$ parties. Actually, the study of asymptotic robustness of GME already connected this property to the ability to distill maximally entangled states between arbitrary pairs of parties \cite{AGME1, AGME2}. Here, we build on the techniques developed therein in order to establish this property for arbitrary multipartite pure states and arbitrary subsets of parties.

This article is organized as follows. In the next section we introduce in mathematical detail all necessary concepts. In Sec.\ III we prove a severe limitation for distillability in star networks. Namely, we prove that if one uses an isotropic star network to teleport a GHZ state, then the resulting state is GHZ-undistillable for any $p<1$ for a sufficiently large subset of parties. Thus, not only it is impossible to purify entanglement for an arbitrary subset of parties from a single copy of the star network, but, even if we stored arbitrarily many copies of the states resulting from teleportation, it is impossible to obtain GHZ states with arbitrarily large fidelity by LOCC. This shows that multipartite distillation protocols are doomed to failure even if we had arbitrarily many copies of the isotropic star network state. Thus, one must have the ability to operate on the many copies of the noisy bipartite entangled links beforehand and apply bipartite distillation protocols before teleportation. In addition to the practical limitation this entails, this result also shows that when we consider later on more connected networks, teleportation and subsequent multipartite distillation cannot work and one is then bound to protocols in which bipartite distillation of the links has to precede teleportation. In Sec.\ IV we present our main results showing that distillation to arbitrary states for arbitrary subsets of parties is possible with a fixed visibility threshold if the graph sequence corresponding to the network has both minimum degree and edge-connectivity growing sufficiently fast. In addition to this, we prove that the latter condition is not only sufficient but also necessary while the former is not necessary. In order to prove the last claim, we explicitly construct a graph sequence with slow degree growth which nevertheless achieves partial distillability. We conclude in Sec.\ V with some discussion and open questions.

\section{Mathematical preliminaries}\label{sec1}

\subsection{Graph theory}

In this subsection we provide the definition of all the graph-theoretic concepts that are used in this work. For a more detailed discussion we refer our readers to the well-established literature in the field such as \cite{bondy,diestel}. All graphs in this article are finite and simple. Thus, a graph is a pair $G=(V,E)$, where $V$ is a finite set of vertices and $E$ a set of edges, consisting of unordered pairs $(u,v)$ of vertices $u,v \in V$ ($u\neq v$). The order of $G$, denoted $|G|$, is the cardinality of the set of vertices $V$ and the size of $G$ is the cardinality of the set of edges $E$.
Two vertices $u,v \in V$ are said to be adjacent if $(u,v) \in E$.
We denote the set of edges adjacent to a vertex $v \in V$ by $E(v)$.
The degree of a vertex $v \in V$, $\delta(v) = |E(v)|$, is the number of adjacent vertices to it, and we denote by $\dmin(G) = \min_{v \in V} \delta(v)$ ($\dmax(G) = \max_{v \in V} \delta(v)$) the minimal (maximal) degree of a graph $G$.
A path is a sequence of edges joining subsequent vertices, $(v_0,v_1),(v_1,v_2), \dots, (v_{\ell-1},v_\ell)$, and when $v_j \neq v_i$ for $i \neq j$ the path is said to be simple. A cycle is a path starting and ending on the same vertex. The length of a path is the number of edges in the path and the distance between two vertices $u,v \in V$, $d_G(u,v)$, is the length of the shortest path joining $u$ and $v$. 

A graph $G$ is said to be connected if for any two vertices there is a path joining them. Intuitively, $\dmin(G)$ and $\dmax(G)$ provide an idea of how connected a given graph $G$ is. However, other notions are possible such as the edge-connectivity of $G$, denoted by $\lambda(G)$, which is the minimum number of edges one needs to remove from $G$ before it becomes disconnected. By Menger's theorem, it also holds that $\lambda(G)=\min_{u,v\in V}\lambda_G(u,v)$, where (for $u\neq v$) $\lambda_G(u,v)$ is the maximal number of edge-disjoint paths joining $u$ and $v$. Notice that for every graph $G$ it holds that $\lambda(G)\leq\dmin(G)\leq|G|-1$, with equality in both cases at the same time only in the case of the completely connected graph (i.e.\  a graph in which every pair of vertices is adjacent and, thus, the most connected graph). Another graph parameter that can be linked to connectivity is the diameter of a graph, which is given by $\mathrm{diam}(G):=\max_{u,v\in V}d_G(u,v)$, and where it should be understood that $\mathrm{diam}(G)=\infty$ if $G$ is not connected.
 
A connected graph without any cycles is called a tree, and a path of degree-two vertices starting and ending in vertices with degree different than two is called a branch. A star graph is a tree in which there is only one vertex with degree bigger than 1, referred to as center. A spider graph (or spider tree, or subdivided star graph) is a tree with at most one vertex with degree bigger than 2. We will call the vertex with degree bigger than 2 the spider's centre (or body) and we call each of the branches starting from it and ending on degree-one vertices its legs.

\subsection{Quantum (network) states and distillability}

In this article we will always consider finite-dimensional Hilbert spaces. A quantum state is given by any positive semidefinite unit-trace operator $\rho$, referred to as density operator, acting on a given Hilbert space $\mathcal{H}$. The set of all density operators on $\mathcal{H}$ is denoted by $D(\mathcal{H})$. When a density operator is a rank-1 projector we say that the corresponding state is pure. When a quantum system is composed by $N$ subsystems, labelled by the elements of $[N]:=\{1,2,\ldots,N\}$, the corresponding Hilbert space is given by $\mathcal{H}=\bigotimes_{k=1}^N\mathcal{H}_k$, where $\mathcal{H}_k$ is the Hilbert space corresponding to party $k$. In the bipartite case, $\mathcal{H}=\mathcal{H}_1\otimes\mathcal{H}_2$, we denote by
\begin{equation}\label{maxent}
|\phi^+\rangle=\frac{1}{\sqrt{d}}\sum_{i=1}^d|ii\rangle,\quad \phi^+=|\phi^+\rangle\langle\phi^+|
\end{equation}
the maximally entangled state, where $d=\dim\mathcal{H}_1=\dim\mathcal{H}_2$, we use the shorthand $|ij\rangle=|i\rangle_1\otimes|j\rangle_2$ and $\{|i\rangle_k\}_{i=1}^d$ denotes the canonical basis of $\mathcal{H}_k$. For any $p\in[0,1]$ the isotropic state is given by
\begin{equation}\label{isotropic}
\rho(p) = p \phi^++  \frac{(1-p)}{d^2} \one,
\end{equation}
where $p$ is referred to as the visibility and $\one$ stands for the identity operator on $\mathcal{H}$. It is known that the isotropic state is entangled if and only if (iff) $p>1/(d+1)$ \cite{isotropicpaper} and it is a common noise model in which one produces a maximally entangled state with probability $p$ and white noise with probability $1-p$.

The manipulation of multipartite states distributed to distant parties is constrained to a particular class of completely positive and trace preserving (CPTP) maps $\Lambda:D(\mathcal{H})\to D(\mathcal{H}')$ known as LOCC maps. These maps capture the fact that the parties can only induce local dynamics correlated by the exchange of classical communication. The precise (and cumbersome) definition of LOCC CPTP maps is not necessary to follow this paper and the interested reader is referred to \cite{locc1,locc2}. We will only use some basic facts such as that the set of LOCC maps is convex and closed under composition, that teleportation (see Sec.\ II.C below) can be implemented by LOCC and that every map whose output is a fixed separable state is LOCC. The identity map on $D(\mathcal{H})$, which we denote by $\id$, is also obviously LOCC. In order to measure the closeness of quantum states we will consider the fidelity, given by
\begin{equation}\label{fidelity}
F(\rho,\sigma) = \tr^2 \sqrt{\sqrt{\rho}\sigma \sqrt{\rho}},
\end{equation} 
for any $\rho,\sigma\in D(\mathcal{H})$. Notice that it always holds that $0\leq F(\rho,\sigma)\leq1$, with the latter equality only if $\rho=\sigma$. A bipartite state $\rho\in D(\mathcal{H})=D(\mathcal{H}_1\otimes\mathcal{H}_2)$ is said to be distillable to the maximally entangled state (or distillable for short) if there exists an LOCC map $\Lambda_k:D(\mathcal{H}^{\otimes k})\to D(\mathcal{H})$ such that 
\begin{equation}
\lim_{k\to\infty}F(\Lambda_k(\rho^{\otimes k}),\phi^+)=1.
\end{equation}
A state is distillable only if it is entangled. In the particular case of isotropic states, they are distillable iff entangled \cite{isotropicpaper}. Even if entangled, states that remain positive semidefinite after partial transposition (or PPT for short) are known not to be distillable \cite{pptundistillable}. Partial transposition is defined to be the linear map $\Gamma:B(\mathcal{H}_1\otimes\mathcal{H}_2)\to B(\mathcal{H}_1\otimes\mathcal{H}_2)$ such that $\Gamma(|ij\rangle\langle kl|)=|kj\rangle\langle il|$.

An $N$-partite quantum network state is described by a graph $G=(V,E)$ such that $|G|=N$. The vertices (i.e.\ the elements of $V$) represent the parties and the edges (i.e.\ the elements of $E$) that the corresponding pair of parties share some bipartite state to be specified. Thus, we denote the total Hilbert space by $\mathcal{H}_G$ and $\mathcal{H}_G = \bigotimes_{v \in V} \mathcal{H}_v$, where $\mathcal{H}_v$ represents the Hilbert space corresponding to the party represented by the vertex $v\in V$. Notice that there is a bijective identification between the elements of $V$ and those of $[N]$ and, thus, in a slight abuse of notation, we will refer to both vertices and parties as elements of these two sets equally. If vertices $u$ and $v$ share the state $\rho$, then we write $\rho_e$, where $e=(u,v)\in E$, and $\rho_e\in D(\mathcal{H}_e)$, where $\mathcal{H}_e=\mathcal{H}_{ue}\otimes\mathcal{H}_{ve}$. Thus, $\mathcal{H}_v=\bigotimes_{e\in E(v)} \mathcal{H}_{ve}$. Furthermore, it also holds that $\mathcal{H}_G=\bigotimes_{e\in E}\mathcal{H}_e$, even though $\mathcal{H}_e$ does not correspond to the local Hilbert space of any party. %Finally, for a given subset of parties $V_0 \subset V$, we define the \emph{inner} edges of $V_0$ as $E(V_0) = \{(v_1,v_2) \in E \mid v_1,v_2 \in V_0\}$ and it will be convenient for later purposes to define the associated Hilbert space $\mathcal{H}_{V_0} = \bigotimes_{e \in E(V_0)} \mathcal{H}_e$.

Isotropic quantum networks correspond to the case in which all parties joined by an edge share an isotropic state with the same visibility $p$ and local dimension $d$ (we will assume the latter parameter to be clear from the context and will not explicitly refer to it). Thus, given a graph $G$ and a visibility $p\in[0,1]$ (and local dimension $d$), the corresponding isotropic network state is given by
\begin{equation}\label{isotropicnetwork}
\sigma(G,p)=\bigotimes_{e\in E}\rho_e(p),
\end{equation}
where $\rho(p)$ is given by Eq.\ (\ref{isotropic}). Therefore, in this case for each party $v\in V$ it holds that $\dim\mathcal{H}_v=d^{\delta(v)}$ and $\dim\mathcal{H}_G=d^{2|E|}$. Notice that if $G$ is not connected (or $p\leq1/(d+1)$), then $\sigma(G,p)$ is not entangled. Hence, in the following we will always assume that all graphs are connected. Since we want to analyze the properties of isotropic network states as the size of the network increases according to a particular pattern, once the visibility is specified all the necessary information is encoded into a graph sequence $(G_n)_{n\in\mathbb{N}}$ (with $G_n=(V_n,E_n)$) that specifies this pattern and such that $N=N(n)=|G_n|\to\infty$ as $n\to\infty$ (we will assume all graph sequences in the following to fulfill this latter condition). Furthermore, we would like to study entanglement concentration for specific subsets of parties as $n\to\infty$; thus, we will always assume that $V_n\subset V_{n+1}$ $\forall n$. Finally, we have to constrain the dimensionality of the arbitrary target states we would like to distill to for a given subset of parties $V$; thus, we set the notation $\mathcal{H}'_V=\bigotimes_{v\in V}\mathcal{H}'_v$ with $\dim\mathcal{H}'_v=d$ $\forall v$ for the output Hilbert spaces. We can now precisely define the property of isotropic network states that we want to analyze as described informally in the introduction and that we term partial distillability.

\begin{definition}\label{partialdist}
We say that a graph sequence $(G_n)_{n\in\mathbb{N}}$ has partial distillability if there exists a fixed value $p_0<1$ such that for any given $m\in\mathbb{N}$, any subset of parties $V_0$ such that $|V_0|=m$ and any pure state $\psi\in D(\mathcal{H}'_{V_0})$, there exists an LOCC map $\Lambda_{n,p}:D(\mathcal{H}_{G_n})\to D(\mathcal{H}'_{V_0})$ such that
\begin{equation}\label{partialdisteq}
\lim_{n\to\infty}F(\Lambda_{n,p}(\sigma(G_n,p)),\psi)=1
\end{equation}
holds for all $p> p_0$.
\end{definition}

Thus, for any given graph sequence the threshold in the visibility $p_0$ can only depend on $d$. Notice that in the above definition whenever $m$ and $V_0$ are chosen, one only considers the subsequence $(G_{n})_{n\geq n_0}$ corresponding to the terms such that $|G_{n}|\geq m$ and $V_0\subset V_{n}$ $\forall n\geq n_0$, which is always ensured by a choice of $n_0$. In order not to make things more complicated, we have not incorporated these details into the definition and subsequent discussions, hoping that this will be clear to the reader.

\subsection{Noisy teleportation} \label{subsec:noisy_teleport}

The main ingredient in the LOCC protocols that we are going to consider is teleportation \cite{teleport,nielsenchuang}. We denote by $\mathcal{T}: D(\mathcal{H}_1 \otimes \mathcal{K}_1 \otimes \mathcal{H}_2) \to D(\mathcal{H}_2)$ the channel implementing the standard bipartite teleportation protocol (where $\dim\mathcal{H}_1=\dim\mathcal{K}_1 =\dim \mathcal{H}_2$ and $\mathcal{H}_1 \otimes \mathcal{K}_1$ is the Hilbert space of one party and $\mathcal{H}_2$ is the Hilbert space of the other). That is, if $\phi^+\in D(\mathcal{K}_1\otimes\mathcal{H}_2)$ denotes the maximally entangled state (cf.\ Eq.\ (\ref{maxent})), then
\begin{equation}
\mathcal{T}(\tau \otimes \phi^+) = \tau
\end{equation}
holds for every choice of input state $\tau\in D(\mathcal{H}_1)$. We will be often using the entangled links in an isotropic network state to teleport quantum states through the network, which will be noisy when the visibility satisfies $p<1$. Applying the same protocol by using an isotropic state instead of a maximally entangled state yields
\begin{equation}
\mathcal{T}(\tau \otimes \rho(p)) = p \tau + (1-p) \frac{\one}{d}.
\end{equation}
We define the noisy teleportation channel $\mathcal{T}_p: D(\mathcal{H}_1) \to D(\mathcal{H}_2)$ as
\begin{equation}
\mathcal{T}_p(\tau):=\mathcal{T}(\tau \otimes \rho(p)),
\end{equation}
which acts on the canonical basis as
\begin{equation}
\mathcal{T}_p(|i \rangle\langle j|) = p|i \rangle\langle j| + (1-p) \delta_{ij} \frac{\one}{d}.
\end{equation}

We will consider later two sequential versions of this noisy teleportation. Notice first that if we use an isotropic state with visibility $p$ to teleport one share of another isotropic state with the same visibility, then we obtain another isotropic state with visibility $p^2$, i.e.\
\begin{equation}
\id\otimes\mathcal{T}_p(\rho(p))=\rho(p^2).
\end{equation}
Thus, the first version corresponds to sequentially teleporting over a path in an isotropic network in order to establish entanglement between the initial and final nodes. Given a path $P = (v_0, v_1) (v_1, v_2) \dots (v_{\ell-1}, v_\ell)$ of length $\ell$ joining $v_0$ with $v_\ell$, each party teleports the part of the state shared with the previous party in the path to the next one. On each teleportation step the received state is a noisier isotropic state, and the full process boils down to an LOCC protocol, $\Lambda_P: \mathcal{H}_P \to \mathcal{H}'_{\{v_0,v_\ell\}}$, satisfying
\begin{equation} \label{eq:teleportation-path}
\Lambda_P (\sigma(P,p)) = \rho_{(v_0, v_\ell)}(p^{2^{\ell-1}}).
\end{equation}
This allows the parties $v_0$ and $v_\ell$ to share a quantum channel via further noisy teleportation provided the length of the path $\ell$ is small enough and the initial visibility $p$ is large enough.

Another option is to directly distribute a multipartite quantum state to some set of parties $V$ by means of a star graph $S=(V^\ast,E)$, where $V^*=V \cup\{v_0\}$ and $v_0$ is its center. This additional party can locally prepare any state $\tau\in D(\mathcal{H}'_V)$ and teleport each share of it to each of the parties in $V$ by means of each edge in the star graph using the corresponding isotropic state as a resource for noisy teleportation. Thus, this corresponds to implementing the LOCC map from $D(\mathcal{H}_S)$ to $D(\mathcal{H}'_V)$ given by $(\mathcal{T}_p)^{\otimes|V|}(\tau)$ for any chosen $\tau\in D(\mathcal{H}'_V)$. If the parties in $V$ are linked to many additional parties giving rise to star graphs with different centers, they can iterate this process in order to end up with many noisy copies of the target state, to which they can apply multipartite distillation protocols as a final stage. We will in particular later on consider the paradigmatic GHZ $n$-qubit state
\begin{align}\label{ghz}
|\ghz\rangle&=\frac{1}{\sqrt{2}}(|0\rangle^{\otimes n}+|1\rangle^{\otimes n}),\nonumber\\ \ghz&=|\ghz\rangle\langle \ghz|,
\end{align}
which is an essential resource in many applications. In this case a lengthy but straightforward calculation yields
\begin{align}
(\mathcal{T}_p)^{\otimes n}(\ghz) &= p^n \ghz \nonumber\\ &+
    \sum_{k=1}^{n-1} p^{n-k} (1-p)^k \sum_{C \in \mathcal{C}(n,k)} \sigma_k^C \nonumber\\ &+ (1-p)^n \frac{\one}{2^n},\label{teleportedghz}
\end{align}
where $\mathcal{C}(n,k)$ represents the set of subsets of $[n]$ with size $k$ and
\begin{equation}\label{sigmak}
    \sigma_k^C = \frac{1}{2}\left(
      \bigotimes_{\ell \notin C} | 0 \rangle_{\ell} \langle 0 | +
      \bigotimes_{\ell \notin C} | 1 \rangle_{\ell} \langle 1 |
    \right) \otimes \frac{\one_C}{2^k}.
\end{equation}

\section{Limitation of multipartite distillation protocols}\label{sec2}

Before addressing the main question of whether partial distillability, as given in Definition 1, is possible and, if so, which graph sequences display it, in this section we establish first a fundamental limitation of the second class of sequential noisy teleportation protocols that we presented at the end of the previous section. This result dooms to failure strategies that are based on multipartite distillation already for the particular paradigmatic case of the GHZ state, carries an important weight for the analysis to follow, and might be of independent practical interest. We begin by stating precisely this claim in the following theorem.

\begin{theorem}\label{th:noghzdistillation}
For any given $p<1$, there exists $n_0\in\mathbb{N}$ such that the state $(\mathcal{T}_p)^{\otimes n}(\ghz)$ (cf.\ Eq.\ (\ref{teleportedghz})) that arises after teleporting a GHZ state in an isotropic star network is not GHZ-distillable $\forall n\geq n_0$.
\end{theorem}

This means that, outside the ideal case of a perfect visibility, for any value this parameter may take, if the number of parties is big enough, then the state resulting from noisy-teleporting the GHZ state is no longer GHZ-distillable. Hence, even if we could store an arbitrary number of copies of the teleported state and act on them with arbitrary LOCC protocols, the obtainable fidelity with the GHZ state remains bounded away from 1. This applies to any non-perfect value of the initial visibility no matter how close to 1, and, therefore, there is no threshold in the visibility for which this purification strategy works for an arbitrary number of parties. In a previous work, a similar obstruction had been found for the question of GME versus biseparability. Namely, any tree isotropic network (and, hence, a star isotropic network) becomes biseparable for any given $p<1$ if its size is big enough \cite{AGME1}. Notice, however, that this does not forbid distillability because GME can be superactivated \cite{GMEactivation1,GMEactivation2} and biseparable states can in principle be distilled to GME states \cite{distillationbiseparable}. 

Theorem 2 has several important consequences. First, this shows a major obstacle for partial distillability for strategies based on multipartite distillation. If the vertices in our set of interest, $V_0$, are connected via arbitrarily many star subgraphs to arbitrarily many different centers, for any given visibility we cannot obtain in this way a GHZ state with fidelity as close to 1 as desired if the size of $V_0$ is big enough even if the number of star subgraphs is unbounded. Thus, partial distillability is impossible in this scenario by teleportation of the GHZ state and subsequent multipartite distillation because the threshold $p_0$ in the visibility would need to depend on $|V_0|$. Second, this is also relevant in the standard quantum network scenario in which the entangled links are constantly replenished and where star configurations are often considered (see e.g.\ \cite{networks3}). This is because if we teleport over the noisy links and rebuild the entangled pairs to reiterate this process arbitrarily many times, we will not be able to distill to a GHZ state with subsequent purification protocols if the size of the network is big enough for any given initial visibility $p<1$. Notice, however, that this does not mean that noisy quantum networks are useless above a certain size. What this shows is that if we want to distribute entangled states to an arbitrarily large number of parties and if the entangled links can be refreshed, then they should be stored so that we can apply bipartite distillation protocols to purify links and then teleport sufficiently noiselessly. Similarly, Theorem 2 does not imply that partial distillability is impossible. Nevertheless, it indicates that, rather than applying multipartite distillation protocols after teleportation, bipartite entanglement needs to be routed in the network so that bipartite distillation precedes the final teleportation step in the line of the first class of sequential teleportation protocols outlined at the end of the previous section.  

Before developing these ideas in the next section, we conclude this one presenting the proof of Theorem~\ref{th:noghzdistillation}. For this, we state and prove first some technical lemmas. Also, from \cite{DurCirac2000}, we introduce the following notation. For $0 \leq j < 2^{n-1}$ and letting $j_1j_2 \dots j_{n-1}$ be the binary expression of $j$ we define the vectors
\begin{equation}
|\Psi_j^{\pm}\rangle = \frac{1}{\sqrt{2}} (|j_1j_2 \dots j_{n-1}\rangle |0\rangle \pm |\overline{j_1} \,\overline{j_2} \dots \overline{j_{n-1}}\rangle |1\rangle),
\end{equation}
which give an orthonormal basis of $(\mathbb{C}^2)^{\otimes n}$ (here and in the following $\overline{j_k}$ denotes the binary negation of $j_k$). Notice that $|\ghz\rangle=|\Psi_0^+ \rangle$.

 % We will use the notation
 % \begin{equation*}
 %   |j\rangle \coloneqq |j_1j_2 \dots j_{n-1}\rangle, \qquad
 %   |\overline{j}\rangle \coloneqq |\overline{j_1} \,\overline{j_2} \dots \overline{j_{n-1}}\rangle.
 % \end{equation*}

 % Since $GHZ = |\Psi_0^+ \rangle\langle \Psi_0^+ |$, it follows
 % \begin{equation} \label{eq:matrix_elements_sigmap}
 %   \langle \Psi_j^{\pm}| \sigma(p) |\Psi_j^{\pm}\rangle = 
 %   p^n \delta_{j0} \delta_{\pm+} + \sum_{k=1}^{n-1} p^{n-k} (1-p)^k \sum_{C \in \mathcal{C}(n,k)} \langle \Psi_j^{\pm}| \sigma_k^C |\Psi_j^{\pm}\rangle + \left(\frac{1-p}{2}\right)^n,
 % \end{equation}
 % with $\delta_{js}$ the Kronecker delta.
 
\begin{lemma} \label{prop:sigma_k_c_eigenvectors}
Let $1\leq k < n$ and $0\leq j < 2^{n-1}$ be integers and $C \in \mathcal{C}(n,k)$. Denote by $j_1j_2 \dots j_{n-1}$ the binary expression of $j$, and define $j_n=0$. Then, $|\Psi_j^{\pm}\rangle$ is an eigenvector of $\sigma_k^C$ as given in Eq.\ (\ref{sigmak}). Moreover, the associated eigenvalue is non-vanishing only if for every $\ell_1,\ell_2 \notin C$ it holds $j_{\ell_1} = j_{\ell_2}$. That is,
\begin{equation}
      \sigma_k^C |\Psi_j^{\pm}\rangle = \begin{cases}
        \dfrac{1}{2^{k+1}}|\Psi_j^{\pm}\rangle & \text{if }j_{\ell_1} = j_{\ell_2} \forall \ell_1,\ell_2 \notin C, \\
        0 & \text{otherwise}.
      \end{cases}
\end{equation}
\end{lemma}
\begin{proof}
With the notation  introduced, we can write
\begin{equation}
      |\Psi_j^\pm\rangle = \frac{1}{\sqrt{2}} \left(
        \bigotimes_{\ell \leq n} |j_\ell\rangle_\ell
        \pm
        \bigotimes_{\ell \leq n} |\overline{j_\ell}\rangle_\ell
      \right).
\end{equation}
It is then clear that
\begin{multline} \label{eq:sigma_kC_psi}
      \sigma_k^C |\Psi_j^\pm\rangle
      \\=  \frac{1}{2^{k+1}\sqrt{2}} \left(
        \bigotimes_{\ell \notin C} \langle 0| j_\ell\rangle |0\rangle_\ell
        + \bigotimes_{\ell \notin C} \langle 1| j_\ell\rangle |1\rangle_\ell
      \right) \bigotimes_{\ell \in C} |j_\ell\rangle_\ell \\
      \pm \frac{1}{2^{k+1}\sqrt{2}} \left(
        \bigotimes_{\ell \notin C} \langle 0| \overline{j_\ell}\rangle |0\rangle_\ell
        + \bigotimes_{\ell \notin C} \langle 1| \overline{j_\ell}\rangle |1\rangle_\ell
      \right)
      \bigotimes_{\ell \in C} |\overline{j_\ell}\rangle_\ell.
\end{multline}
Therefore, if there are $\ell_1, \ell_2 \notin C$ such that $j_{\ell_1} \neq j_{\ell_2}$ it follows $\sigma_k^C |\Psi_j^\pm\rangle = 0$.
That is, $\sigma_k^C |\Psi_j^\pm\rangle$ can only be non-vanishing if every $j_\ell$ with $\ell \notin C$ have the same value.
If $\ell \notin C$ implies $j_\ell = 0$, then Eq.~\eqref{eq:sigma_kC_psi} implies
\begin{align}
      \sigma_k^C |\Psi_j^\pm\rangle
      &= \frac{1}{2^{k+1}\sqrt{2}} \left(
        \bigotimes_{\ell \notin C} |0\rangle_\ell \bigotimes_{\ell \in C} |j_\ell\rangle_\ell
        \pm \bigotimes_{\ell \notin C} |1\rangle_\ell \bigotimes_{\ell \in C} |\overline{j_\ell}\rangle_\ell
      \right)\nonumber\\
      &= \frac{1}{2^{k+1}} |\Psi_j^\pm\rangle.
\end{align}
Analogously, if $\ell \notin C$ implies $j_\ell = 1$ then $\sigma_k^C |\Psi_j^\pm\rangle = |\Psi_j^\pm\rangle / 2^{k+1}.$
\end{proof}

\begin{lemma}\label{lemmasj}
For $0\leq j < 2^{n-1}$ an integer, define
\begin{equation}\label{sj}
      S_j := \sum_{k=1}^{n-1} p^{n-k} (1-p)^k \sum_{C \in \mathcal{C}(n,k)} \langle \Psi_j^{\pm}| \sigma_k^C |\Psi_j^{\pm}\rangle.
\end{equation}
%\begin{equation}
% \langle \Psi_j^\pm | (\mathcal{T}_p)^{\otimes n}(\ghz) | \Psi_j^\pm \rangle
%        = \frac{1}{2} \left[\pm\delta_{j0}p^n+ \left(\frac{1+p}{2}\right)^n \left(\frac{1-p}{1+p}\right)^{w_H(j)} + \left(\frac{1-p}{2}\right)^n \left(\frac{1+p}{1-p}\right)^{w_H(j)} \right],
%\end{equation}
Then,
\begin{equation}
S_0= \frac{1}{2} \left[ -p^n + \left(\frac{1+p}{2}\right)^n - \left(\frac{1-p}{2}\right)^n \right],
\end{equation}
and for $j\neq0$
    \begin{align}
      S_j &= \frac{1}{2} \left[ \left(\frac{1+p}{2}\right)^n \left(\frac{1-p}{1+p}\right)^{w_H(j)} \right.\nonumber\\
      &+ \left.\left(\frac{1-p}{2}\right)^n \left(\frac{1+p}{1-p}\right)^{w_H(j)} \right] - \left(\frac{1-p}{2}\right)^n,\label{calculationsj}
    \end{align}
where $w_H(j)$ is the Hamming weight of $j$ (i.e.\ the number of 1s in its binary representation).
\end{lemma}
%\end{widetext}
\begin{proof}
%\begin{widetext}
From Lemma \ref{prop:sigma_k_c_eigenvectors} it follows that
%  \begin{align} \label{eq:matrix_elements_sigmap}
%    \langle \Psi_j^{\pm}| (\mathcal{T}_p)^{\otimes n}(\ghz) |\Psi_j^{\pm}\rangle &= 
%    \left(\frac{1\pm1}{2}\right) p^n \delta_{j0} + \sum_{k=1}^{n-1} p^{n-k} (1-p)^k \sum_{C \in \mathcal{C}(n,k)} \langle \Psi_j^{\pm}| \sigma_k^C |\Psi_j^{\pm}\rangle + \left(\frac{1-p}{2}\right)^n\nonumber\\
%     &= 
%    \left(\frac{1\pm1}{2}\right) p^n \delta_{j0} + \frac{1}{2} \sum_{k=1}^{n-1} \#(j,k) p^{n-k} \left(\frac{1-p}{2}\right)^k + \left(\frac{1-p}{2}\right)^n,
%  \end{align}
%\end{widetext}
\begin{equation}\label{sj2}
      S_j = \frac{1}{2} \sum_{k=1}^{n-1} \#(j,k) p^{n-k} \left(\frac{1-p}{2}\right)^k,
\end{equation}
where $\#(j,k)$ is the number of elements in $\{C \in \mathcal{C}(n,k) \mid \langle \Psi_j^{\pm}| \sigma_k^C |\Psi_j^{\pm}\rangle \neq 0\}$; that is, the number of non-vanishing terms in the sum $\sum_{C \in \mathcal{C}(n,k)} \langle \Psi_j^{\pm}| \sigma_k^C |\Psi_j^{\pm}\rangle$. Let $C \in \mathcal{C}(n,k)$, $j_1 j_2 \dots j_{n-1}$ the binary expression of $j$ and $j_n = 0$ in order to use the same notation as above.
Given $j$, Lemma~\ref{prop:sigma_k_c_eigenvectors} tells us $\#(j,k)$ is the number of elements $C \in \mathcal{C}(n,k)$ such that for every $\ell_1, \ell_2 \notin C$ it holds $j_{\ell_1} = j_{\ell_2}$. For $i=0,1$, let us denote by $\#_i(j,k)$ the number of elements $C \in \mathcal{C}(n,k)$ such that $j_\ell = i$ implies $\ell \in C$. By definition, it holds $\#(j,k) = \#_1(j,k) + \#_0(j,k)$. It follows now that $\#_1(j,k) = 0$ whenever $k < w_H(j_1 j_2 \dots j_{n-1}0)=w_H(j)$. Otherwise, $\#_1(j,k)$ must be the number of different subsets of $[n]$, with size $k$, containing the positions in which $j_1 j_2 \dots j_{n-1}0$ has ones. Since these $w_H(j)$ positions have to be assigned to $C$ always, $\#_1(j,k)$ is given by the different ways in which the remaining $n-w_H(j)$ positions can fill the remaining $k-w_H(j)$ slots in $C$, i.e.\
\begin{equation}
      \#_1(j,k) = \binom{n-w_H(j)}{k-w_H(j)}.
\end{equation}
On the other hand, for $\#_0(j,k)$ the $n-w_H(j)$ positions where $j_1 j_2 \dots j_{n-1}0$ has a zero must always be assigned to $C$ and, therefore, $\#_0(j,k)$ is given by the different ways in which the remaining $w_H(j)$ positions can fill the remaining $k-(n-w_H(j))$ slots in $C$, i.e.\
\begin{equation}
      \#_0(j,k) = \binom{w_H(j)}{k + w_H(j) - n}.
\end{equation}
Using the results in the previous discussion, for $j \neq 0$, it follows that
\begin{align}
      S_j &=
      \frac{1}{2} \sum_{k=w_H(j)}^{n-1} \binom{n-w_H(j)}{k-w_H(j)} p^{n-k} \left(\frac{1-p}{2}\right)^k\nonumber\\
      &+ \frac{1}{2} \sum_{k=n-w_H(j)}^{n-1} \binom{w_H(j)}{k + w_H(j) - n} p^{n-k} \left(\frac{1-p}{2}\right)^k,
\end{align}
while
\begin{equation}
      S_0 = \frac{1}{2} \sum_{k=1}^{n-1} \binom{n}{k} p^{n-k} \left(\frac{1-p}{2}\right)^k.
\end{equation}
All these sums can be calculated using the binomial theorem by shifting the sum index appropriately, yielding the desired result.
\end{proof}

\begin{lemma}\label{lemma:ppt}
Let $M$ be a fixed subset of parties. Then, for any fixed $p<1$ there is $n_0\in\mathbb{N}$ such that $(\mathcal{T}_p)^{\otimes n}(\ghz)$ is PPT in the bipartition $M|\overline{M}$ for all $n\geq n_0$.
\end{lemma}
\begin{proof}
Given $M$, let $n$ be large enough so that $M\subsetneq[n]$. Without loss of generality we can assume that $n\notin M$ and let $r$ be the integer in $(0,2^{n-1})$ such that in its binary representation $r_1 r_2 \dots r_{n-1}$ it holds that $r_j=1$ if $j\in M$ and $r_j=0$ otherwise. Notice that
\begin{align}
\Gamma_{M}[(\mathcal{T}_p)^{\otimes n}(\ghz)] &= p^n \Gamma_{M}(\ghz) \nonumber\\
&+ \sum_{k=1}^{n-1} p^{n-k} (1-p)^k \sum_{C \in \mathcal{C}(n,k)} \sigma_k^C\nonumber \\ 
&+ (1-p)^n \frac{\one}{2^n},
\end{align}
with
\begin{align}
\Gamma_{M}(\ghz) &= \frac{1}{2} \left(|0 \dots 0 \rangle\langle 0 \dots 0 | + |1 \dots 1 \rangle\langle 1 \dots 1 |\right.\nonumber \\ 
&\left.+ |r0\rangle\langle \overline{r}1| + |\overline{r}1\rangle\langle r0|\right).
\end{align}
Thus, $\{|\Psi_j^{\pm}\rangle \}$ is also an orthonomal basis of eigenvectors of $\Gamma_{M}[(\mathcal{T}_p)^{\otimes n}(\ghz)]$ and
\begin{align}
\langle \Psi_j^\pm | \Gamma_{M}[(\mathcal{T}_p)^{\otimes n}(\ghz)] | \Psi_j^\pm \rangle
      &= \frac{p^n}{2}(\delta_{j0}\pm\delta_{jr})\nonumber\\ &+S_j+ \left(\frac{1-p}{2}\right)^n,
\end{align}
where $S_j$ is given by Eq.\ (\ref{sj}). Therefore, using Lemma~\ref{lemmasj} we obtain that
\begin{widetext}
\begin{equation}
      \langle \Psi_j^\pm | \Gamma_{M}[(\mathcal{T}_p)^{\otimes n}(\ghz)] | \Psi_j^\pm \rangle
      = \frac{p^n}{2} \left[\pm \delta_{jr} + \left(\frac{1+p}{2p}\right)^n \left(\frac{1-p}{1+p}\right)^{w_H(j)} + \left(\frac{1-p}{2p}\right)^n \left(\frac{1+p}{1-p}\right)^{w_H(j)}\right].
\end{equation}
\end{widetext}
Consequently, all eigenvalues of $\Gamma_{M}[(\mathcal{T}_p)^{\otimes n}(\ghz)]$ are positive with the only possible exception of the eigenvalue associated with $|\Psi_r^-\rangle$. However, since $w_H(r)=|M|$ is fixed and $(1+p)/(2p)>1$ for any $p\in(0,1)$, we have that for any given $p$ and sufficiently large $n$
\begin{equation}
\left(\frac{1+p}{2p}\right)^n \left(\frac{1-p}{1+p}\right)^{w_H(r)}>1,
\end{equation}
making the eigenvalue corresponding to $|\Psi_r^-\rangle$ positive as well. This proves the claim.
\end{proof}

\begin{proof}[Proof of Theorem~\ref{th:noghzdistillation}]
For any given pair of parties, the state $\ghz$ can be transformed by LOCC into a bipartite maximally entangled state (see e.g.\ \cite{DurCirac2000}). Thus, by continuity of linear maps and the fidelity, if we assume that $(\mathcal{T}_p)^{\otimes n}(\ghz)$ is GHZ-distillable for $n$ sufficiently large for some given $p<1$, it follows that $(\mathcal{T}_p)^{\otimes n}(\ghz)$ cannot be PPT for $n$ sufficiently large for some given $p<1$ in any bipartion. This is, however, in contradiction with Lemma~\ref{lemma:ppt}.
\end{proof}

\section{Partial distillabillity for sufficiently connected networks}\label{sec3}

In this section we move on to study the main question posed in this work. We begin by showing that partial distillability is indeed possible in isotropic quantum networks if the connectivity grows sufficiently fast as the network becomes bigger. To wit, it is sufficient that both the minimal degree and the edge-connectivity of the associated graph sequence $(G_n)_{n\in\mathbb{N}}$ grow with $n$, provided the former does so fast enough, i.e.\ linearly with the order. This is the content of Theorem \ref{th:sufficient} below. After stating and proving this theorem, we consider in turn necessary conditions for partial distillability in the next subsection. Theorem \ref{th:necessary} proves that in fact $\lambda(G_n)$ must diverge in order to have partial distillability. While, given that $\dmin(G)\geq\lambda(G)$ for every graph $G$, this implies that the minimal degree must diverge as well, this still leaves open whether the linear growth condition of Theorem \ref{th:sufficient} happens to be necessary as well or whether partial distillability is possible with a slower degree growth. We close this section by showing in Theorem \ref{th:slow} that the latter case holds. In fact, we prove therein that partial distillabillity is possible with degree growth at any fractional power of choice by explicitly constructing graph sequences with these properties. This is particularly relevant because the minimal degree can be linked to the overall cost of preparing the network. It provides a lower bound to the number of entangled links that need to be established and it gives the minimal local dimension of the isotropic network state and, therefore, the minimal local Hilbert space each party has to support and control. In this sense, the particular constructions used in Theorem \ref{th:slow} provide according to this figure of merit the cheapest instances of networks with partial distillability that we have been able to find.

\subsection{A sufficient condition for partial distillability}\label{sec3suf}

\begin{theorem}\label{th:sufficient}
If $(G_n)_{n\in\mathbb{N}}$ is a graph sequence such that $\dmin(G_n)$ is $\Omega(|G_n|)$ and $\lambda(G_n)$ is $\omega(1)$ (i.e.\ $\lim_{n\to\infty}\lambda(G_n)=\infty$), then it has partial distillability.
\end{theorem}

The main idea behind the proof of this theorem is taken from \cite{AGME1,AGME2} using the sequential noisy teleportation procols described in Sec.\ \ref{subsec:noisy_teleport}. Given the set $V_0$ of $m$ parties where we aim to distill an $m$-partite state $\psi$, we will show that, for some choice of $v_0\in V_0$, $G_n$ contains a growing number of edge-disjoint spider subgraphs all of them having $v_0$ at its center and the remaining vertices in $V_0$ at the end points of its legs. Then, the isotropic edges can be sequentially teleported from the center through the legs of the spiders so that at the end $v_0$ shares a growing number of noisier isotropic states with each of the other elements in $V_0$. If all the isotropic states in each of these chunks have visibility large enough, then they can be distilled to a maximally entangled state with fidelity arbitrarily close to 1 as $n$ grows. The key observation is that this is ensured for some fixed $p<1$ if the lengths of all legs of all spiders are uniformly bounded. Then, $v_0$ can locally prepare $\psi$ and use the states obtained by distillation to teleport it to the remaining parties with fidelity arbitrarily close to 1 as $n$ grows. 

We move now to the proof of Theorem \ref{th:sufficient}, for which we establish first the following lemma that proves the existence of the spider subgraphs by generalizing the ideas from Proposition~5.2 in \cite{AGME2}.

\begin{lemma} \label{lemma:spider-existence}
Let $G=(V,E)$ be a connected graph with $\dmin(G)>1$ and let $c=\dmin(G)/|G|\in(0,1)$. For every subset $V_0 \subset V$ and any vertex $v \in V_0$, there exist at least $\lfloor \min\{\dmin(G), c\lambda(G)\}/5|V_0|\rfloor$ edge-disjoint spider subgraphs centred in $v$, and having a leg ending on each vertex in $V_0\setminus\{v\}$, each with length upper bounded by $5/c$.
\end{lemma}

\begin{proof}
Throughout the proof we use the result of \cite{ErdosPachPollackEtAl1989} that if $G$ is connected and $\dmin(G)>1$, then
\begin{equation}\label{eq:diameter-bound}
\diam(G) \leq \frac{3|G|}{\dmin(G) + 1} - 1
\leq \frac{3|G|}{\dmin(G)},
\end{equation}
where we use the right-most weaker estimate to simplify calculations. Let $v_1, v_2, \dots, v_{|V_0| - 1}$ be a labelling of the vertices in $V_0 \setminus \{v\}$ such that the sequence of distances $\{d_G(v, v_k)\}$ is non-decreasing.
We denote
\begin{equation}\label{dfinitionM}
M := \left\lfloor \frac{\min\{\dmin(G), c\lambda(G)\}}{5|V_0|} \right\rfloor,
\end{equation}
and we assume that $M\geq1$ since otherwise the claim is trivially true. For each $1 \leq m \leq M$, we will construct the spider graph $S_m$ out of $|V_0| - 1$ paths $P_{m,k}$, $1 \leq k \leq |V_0| - 1$ joining $v$ with $v_k$. To find these paths we use an inductive construction based on a family of subgraphs $G_{m,k}$, with $G_{1,1}=G$. If $G_{m,k}$ is connected, we define $P_{m,k}$ to be the shortest path inside $G_{m,k}$ joining $v$ with $v_k$. We then define $G_{m,k+1}$ to be the subgraph of $G_{m,k}$ obtained by removing all the edges in $P_{m,k}$. Once the last leg of the spider is obtained (i.e., the path $P_{m,|V_0| - 1}$), we define $S_m$ to be the union of every path in $\{P_{m,k}\}_{k=1}^{|V_0|-1}$ and start building the next spider graph, starting over from $G_{m+1,1} = G_{m,|V_0|}$. Notice then that by construction the spider graphs are guaranteed to be edge-disjoint. 

For the previous construction to work and to prove the claim, we need to ensure for all $1 \leq m \leq M$ and all $1 \leq k \leq |V_0| - 1$ that $G_{m,k}$ is connected and that the length of the path $P_{m,k}$ is upper bounded by $5/c$. We know that $G_{1,1}=G$ is connected by hypothesis, and, by Eq.\ (\ref{eq:diameter-bound}), that $\diam(G_{1,1})\leq3/c<5/c$, so that $P_{1,1}$ fulfills the desired bound on its length. Thus, we can use induction over $j$, where we use the bijective mapping $j(k,m):=(m-1)|V_0|+k$, assuming that $1\leq j\leq M|V_0|-1$ (so that $1 \leq m \leq M$ and $1 \leq k \leq |V_0| - 1$). We have just verified that the first element in this set satisfies our claim, so we now assume that $G_{m,k}$ is connected and that the length of $P_{m,k}$ is upper bounded by $5/c$ and we need to show that the same holds for the graph indexed by $j(k,m)+1$ as long as $j\leq M|V_0|-2$. Note that $G_{m,k}$ is built by removing a simple path from $G_{m,k-1}$ (or $G_{m-1,|V_0|-1}$ if $k=1$), and therefore in each step the minimal degree of $G_{m,k}$ can at most be decreased by 2 from the previous step, which implies
\begin{align}
\dmin(G_{m,k}) &\geq \dmin(G) - 2(m-1)|V_0| - 2k +2\nonumber\\
&=\dmin(G)-2j+2\nonumber\\
&\geq \dmin(G) - 2M|V_0|+6\nonumber\\
&\geq\frac{3}{5}\dmin(G)+6\geq6,\label{dminGmk}
\end{align}
where we have used that $j\leq M|V_0|-2$ and Eq.\ (\ref{dfinitionM}). Hence, by the hypothesis that $G_{m,k}$ is connected and the above inequality, we can apply the bound in Eq.~\eqref{eq:diameter-bound} to  $G_{m,k}$ and obtain
\begin{equation}\label{diamGmk}
\diam(G_{m,k})\leq \frac{3|G_{m,k}|}{\dmin(G_{m,k})} \leq \frac{3|G|}{\frac{3}{5}\dmin(G)+4} < \frac{5}{c}.
\end{equation}
Since the diameter of all graphs indexed by any value smaller than $j(k,m)$ can only be smaller than $\diam(G_{m,k})$, the number of edges removed from $G$ in order to obtain $G_{m,k+1}$ must be smaller than
\begin{equation}
        m (|V_0| - 1) \diam(G_{m,k}) < \frac{5m|V_0|}{c} \leq \lambda(G),
\end{equation}
where we have used $5m|V_0| \leq c \lambda(G)$, which follows from the definition of $M$. Thus, since the amount of edges removed from $G$ is smaller than its edge-connectivity, $G_{m,k+1}$ must be connected. From Eq.\ (\ref{dminGmk}) we find that 
\begin{equation}
\dmin(G_{m,k+1})\geq\frac{3}{5}\dmin(G)+4\geq4,
\end{equation}
and, therefore, we can apply Eq.~\eqref{eq:diameter-bound} and, as in Eq.\ (\ref{diamGmk}), we obtain that $\diam(G_{m,k+1})< 5/c$. Hence, $P_{m,k+1}$ must have a length smaller than $5/c$. This finishes the proof
\end{proof}

    \begin{proof}[Proof of Theorem~\ref{th:sufficient}]

%    \begin{proposition} \label{prop:better-isotropic}
%      Let $\{G_n=(V_n, E_n)\}$ be a growing graph sequence such that its minimal degree, $\dmin(G_n)$, is $\Omega(|G_n|)$ and its edge connectivity, $\lambda(G_n)$, diverges.
%      There is a threshold $p_0$ such that for every $p>p_0$, any $p'>0$ and $N>1$ there is $n_0>0$ and an LOCC channel satisfying
%      \begin{equation*}
%        \Lambda(\pi_{G_n}(p)) = \bigotimes_{v \in V_0 \smallsetminus \{v_0\}} \rho_{(v_0, v)}(p')
%      \end{equation*}
%      for all $n>n_0$, $V_0 \in \{V \subset V_{n_0} \mid |V_0| = N\}$ and $v_0 \in V_0$.
%    \end{proposition}
%    \begin{proof}
We will explicitly build a protocol that achieves partial distillation under the given premises. Since $\dmin(G) \leq |G| - 1$ for any graph $G$, $\dmin(G_n)$ being $\Omega(|G_n|)$ implies that there exists $0 < c < 1$ and $n_1 > 0$ such that $\dmin(G_n) \geq c |G_n|$ for every $n > n_1$. In the remainder of the proof we assume that $n$ is large enough so that the above inequality holds. Fix
\begin{equation}
        p_0 = \left(d+1\right)^{- \frac{1}{2^{5/c-1}}}
\end{equation}
and, given $V_0$ of a given fixed size, let
\begin{equation}\label{mnthsuf}
        M_n := \left\lfloor \frac{\min\{\dmin(G_n), c\lambda(G_n)\}}{5|V_0|} \right\rfloor.
\end{equation}
Note then that by our assumptions it holds $\lim_{n \to \infty} M_n = \infty$. By Lemma~\ref{lemma:spider-existence}, for every $V_0 \subset V_n$ and any choice of $v_0 \in V_0$ there are at least $M_n$ spider subgraphs $S_1, \dots, S_{M_n} \subset G_n$ with centre $v_0$ and $|V_0|-1$ legs with length bounded by $5/c$ ending on each of the vertices in $V_0 \setminus \{v_0\}$. Denote in the following by $S_0$ the star graph with vertices given by $V_0$ and $v_0$ as its center. Applying the channel defined in Eq.~\eqref{eq:teleportation-path} to each of the legs of the spider graph $S_j$ allows to transform the isotropic subnetwork state in the spider graph to the isotropic network state in $S_0$ given by
\begin{equation}
        \bigotimes_{v \in V_0 \setminus \{v_0\}} \rho_{(v_0, v)}\left(p^{2^{\ell_j(v)-1}}\right),
\end{equation}
where $\ell_j(v)$ is the length of the path joining $v_0$ with $v$ in $S_j$. Moreover, it is always possible to decrease the visibility of an isotropic state by means of LOCC operations (by mixing the identity channel with a channel that has the normalized identity as fixed output), and so there is an LOCC channel $\Lambda_{S_j}:D(\mathcal{H}_{S_j})\to D(\mathcal{H}_{S_0})$ satisfying
      \begin{equation}
        \Lambda_{S_j} [\sigma(S_j,p)] = \bigotimes_{v \in V_0 \setminus \{v_0\}} \rho_{(v_0, v)}(p^{2^{5/c-1}}).
      \end{equation}
Applying sequentially these channels for all the spider graphs given by Lemma~\ref{lemma:spider-existence} (i.e.\ for $j=1$ up to $j=M_n$) one gets the state
\begin{equation}
        \bigotimes_{v \in V_0 \setminus \{v_0\}} \rho_{(v_0, v)}(p^{2^{5/c-1}})^{\otimes M_n}.
\end{equation}
For $p > p_0$ it holds that $p^{2^{5/c-1}} > 1/(d+1)$, and, since $M_n$ is a growing sequence, the isotropic states can be distilled. Furthermore, the approximating sequence to the maximally entangled state can consist of isotropic states of the same dimension and increasing visibility \cite{isotropicpaper}. Thus, there is an LOCC map $\tilde{\Lambda}_{n,p}:D(\mathcal{H}_{G_n})\to D(\mathcal{H}_{S_0})$ such that 
\begin{equation}\label{distillthsuf}
\tilde{\Lambda}_{n,p}[\sigma(G_n,p)]=\sigma(S_0,p'(n,p)),%\bigotimes_{v \in V_0 \smallsetminus \{v_0\}}\tilde{\rho}_{(v_0, v)}(n,p)
\end{equation}
and, for $p>p_0$, $\lim_{n\to\infty}p'(n,p)=1$. The protocol is finished by $v_0$ preparing locally the desired pure state $\psi$ and teleporting it to the remaining parties in $V_0$ using $\sigma(S_0,p'(n,p))$, i.e.\ by implementing the LOCC map $\id\otimes(\mathcal{T}_{p'(n,p)})^{\otimes(|V_0|-1)}$. That is, using the LOCC map $\tilde{\Lambda}_{n,p}$ defined above, the LOCC map $\mathcal{L}_{v_0}:D(\mathcal{H}_{S_0})\to D(\mathcal{H}_{0})\otimes D(\mathcal{H}_{S_0})$ such that $\mathcal{L}_{v_0}(\rho)=\psi\otimes\rho$ $\forall\rho\in D(\mathcal{H}_{S_0})$ where $\mathcal{H}_{0}=(\mathbb{C}^d)^{\otimes|V_0|}$ is held by party $v_0$, and the teleportation channel $\id\otimes\mathcal{T}^{\otimes(|V_0|-1)}: D(\mathcal{H}_{0})\otimes D(\mathcal{H}_{S_0})\to D(\mathcal{H}'_{V_0})$, the map $\Lambda_{n,p}$ in Definition \ref{partialdist} is given by
\begin{equation}\label{teleportthsuf}
\Lambda_{n,p}=[\id\otimes\mathcal{T}^{\otimes(|V_0|-1)}]\circ\mathcal{L}_{v_0}\circ\tilde{\Lambda}_{n,p}
\end{equation}
and Eq.\ (\ref{partialdisteq}) is guaranteed to hold by the continuity of the fidelity and linear maps.
\end{proof}

\subsection{A necessary condition for partial distillability}\label{sec3nec}

\begin{theorem}\label{th:necessary}
If the graph sequence $(G_n)_{n\in\mathbb{N}}$ has partial distillability, then $\lim_{n\to\infty}\lambda(G_n)=\infty$.
\end{theorem}
\begin{proof}
If the graph sequence $(G_n)_{n\in\mathbb{N}}$ has partial distillability, then, in particular, we can distill bipartite maximally entangled states between any pair of parties as $n\to\infty$. Consider in particular the pair of parties given by $u_1,u_2\in V_n$ for $n$ sufficiently large, let $U_n\subset V_n$ be any such sequence of subsets for which $u_1\in U_n$ and $u_2\in\overline{U_n}$ and let $\partial E(U_n)$ denote the edge-boundary of $U_n$, i.e.\ all edges in $E_n$ that link a vertex in $U_n$ with another vertex in its complement. If we now view the state sequence $\{\sigma(G_n,p)\}$ as a bipartite state in the bipartition $U_n|\overline{U_n}$, our premise implies that for some $p<1$ there exists a bipartite LOCC map $\Lambda_n$ such that 
\begin{equation}
\lim_{n\to\infty}F(\Lambda_n(\rho(p)^{\otimes |\partial E(U_n)|}),\phi^+)=1.
\end{equation}
But $\rho(p)$ is full rank when $p<1$ and, therefore by the resuts of \cite{fang}, it follows that $\lim_{n\to\infty}|\partial E(U_n)|=\infty$. Since this argument applies to any pair of vertices and, hence, to any nontrivial bipartition $U_n|\overline{U_n}$, we arrive at the desired conclusion using that $\lambda(G_n)=\min\{|\partial E(U_n)|:\emptyset\neq U_n\subsetneq V_n\}$.
\end{proof}

Since $\dmin(G)\geq\lambda(G)$ for every graph $G$, the above result automatically implies the following corollary.

\begin{corollary}%\label{th:necessary}
If the graph sequence $(G_n)_{n\in\mathbb{N}}$ has partial distillability, then $\lim_{n\to\infty}\dmin(G_n)=\infty$.
\end{corollary}

\subsection{Graph sequences with partial distillability and sublinear degree growth}\label{sec3slow}

The results of the previous subsection show that both the minimal degree and the edge-connectivity must diverge in order to have partial distillability, while Theorem \ref{th:sufficient} requires additionally that the minimal degree grows linearly with the order of the graphs to guarantee this property. The following theorem shows that this condition is sufficient but not necessary. To do so, we provide moreover explicit partially distillable graph sequences with relatively slow degree growth. These constructions are taken from Theorem~7.2 in \cite{AGME2}, where partial distillability for arbitrary pairs of parties was proven in order to obtain asymptotically robust GME. Here, we generalize the arguments therein so as to demonstrate partial distillabillity for arbitrary subsets of parties.

\begin{theorem}\label{th:slow}
For any choice of $\alpha\in(0,1)$ there is a sequence of graphs $(G_n)_{n \in \mathbb{N}}$ with partial distillability where $\delta_\mathrm{max}(G_n)$  is $O(|G_n|^\alpha)$.
\end{theorem}

\begin{proof}
We construct explicitly the desired graph sequences and we then show that they have partial distillability. Given an integer $k$, define $G_n = (V_n, E_n)$ of order $|G_n| = n^k$ by
\begin{equation}
        V_n = \{I = (i_1, i_2, \dots, i_k) \mid i_m \in [n],\, m \in [k] \}
\end{equation}
and for $I,J \in V_n$, $(I,J) \in E_n$ if $i_m = j_m$ for all but one value of $m$. It follows that all graphs in the sequence are regular and, in particular, for every $v\in V_n$ it holds that
\begin{equation}
        \delta(v) = k(n-1) = k(|G_n|^{1/k} - 1).
\end{equation}
Therefore, for any $\alpha\in(0,1)$, $k$ can be chosen so that $\delta_\mathrm{max}(G_n)$ is $O(|G_n|^\alpha)$. In the same spirit as in Sec.\ \ref{sec3suf}, we will show in the following that for any $V_0 \subset V_n$ and any choice of $v_0 \in V_0$, there are at least 
\begin{equation}\label{numberspidergraphsthslow}
\left(\left\lfloor \frac{n-1}{|V_0| - 1}\right\rfloor - 1\right)k
\end{equation}
edge-disjoint spider subgraphs in $G_n$ with centre at $v_0$ and $|V_0| - 1$ legs, each one ending on a different vertex in $V_0 \setminus \{v_0\}$ and with length upper bounded by $k+1$. This concludes the proof because the number of spider subgraphs grows with $n$ while we can set an upper bound on their size that is independent of $n$. Thus, repeating the exact same arguments as in the proof of Theorem \ref{th:sufficient} we obtain partial distillability.

Given $V_0$, we construct the spider graphs in the following way. For $0 \leq j < |V_0|$, we denote $v_j = (i_{j1}, i_{j2}, \dots, i_{jk})$ the vertices in $V_0$ and for $m \in [k]$ let $\mathcal{I}_m = \{i_{jm} \mid 0 \leq j < |V_0|\}$.
Note that the cardinality of $\mathcal{I}_m$ must be bounded above by $|V_0|$, and so the complement $\mathcal{I}_m^c = [n] \setminus \mathcal{I}_m$ has at least $n - |V_0|$ elements. If $n>|V_0|$ we can construct paths $\{P_j\}_{j=1}^{|V_0| - 1}$, each joining $v_0$ to $v_j$ without passing through any other vertex in $V_0$, by choosing $r_j^{(1)} \in \mathcal{I}_1^c$ and considering the chain of vertices
\begin{widetext}
\begin{equation}
(i_{01}, i_{02}, \dots, i_{0k}) \to (r_j^{(1)}, i_{02}, \dots, i_{0k}) \to (r_j^{(1)}, i_{j2}, i_{03}, \dots, i_{0k})
\to \dots \to (r_j^{(1)}, i_{j2}, \dots, i_{jk}) \to (i_{j1}, i_{j2}, \dots, i_{jk}).
\end{equation}
\end{widetext}
This chain will have repeated consecutive vertices if $i_{0m} = i_{jm}$ for some $m \neq 1$. Removing the repeated vertices from the previous chain we obtain the path $P_j$, whose length is at most $k+1$ by construction. Moreover, if $n \geq 2|V_0|-1$ then $|\mathcal{I}_m^c| \geq |V_0|-1$ $\forall m$, which allows one to take $r_j^{(1)} \neq r_{j'}^{(1)}$ for $j \neq j'$, enforcing all the paths $\{P_j\}_{j=1}^{|V_0| - 1}$ to be disjoint. Therefore, for $n$ large enough $G_n$ has $S_1 = \bigcup_{j=1}^{|V_0| - 1} P_j$ as a spider subgraph with uniformly bounded size, centre at $v_0$ and a leg ending on each vertex of $V_0 \setminus \{v_0\}$.

The construction before can be repeated for any $m \in [k]$ by taking $r_j^{(m)} \in \mathcal{I}_m^c$ and using now the chain
\begin{widetext}
\begin{equation}
        v_0 \to (i_{01}, \dots, i_{0m-1}, r_j^{(m)}, i_{0m+1}, \dots, i_{0k})
        \to \dots \to (i_{j1}, \dots, i_{jm-1}, r_j^{(m)}, i_{jm+1}, \dots, i_{jk})
        \to v_j,
\end{equation}
\end{widetext}
which yields a family $\{S_m\}_{m=1}^k$ of edge-disjoint spider graphs with size bounded by $(|V_0| - 1)(k+1)$.

Now, if $n \geq \nu |V_0|-\nu+1$ for some $\nu\in\mathbb{N}$, then $|\mathcal{I}_m^c| \geq (\nu-1)(|V_0|-1)$ $\forall m$. Thus, the construction above can be repeated by allowing the $\{r_j^{(m)}\}$ to take more values and obtain at least $(\nu - 1)k$ edge-disjoint spider graphs with size at most $(|V_0| - 1)(k+1)$. This finishes the proof.
    \end{proof}

\section{Conclusions}\label{sec4}

In this paper we have introduced the problem of partial distillability in quantum networks of growing size. We have shown that if isotropic states with visibility above a certain fixed threshold are distributed, then any state can be obtained by LOCC for any subset of parties of any given cardinality with fidelity arbitrarily close to 1 as the network increases provided the graph sequence $(G_n)_{n\in\mathbb{N}}$ that describes the network configuration is sufficiently connected. In more detail, we have shown that it is necessary that both the minimal degree and the edge-connectivity grow unboundedly with $|G_n|$ and that it is sufficient if additionally the minimal degree grows linearly with $|G_n|$. Furthermore, we have proven that this latter condition is not necessary by providing instances of graph sequences with partial distillabillity such that the degree grows sublinearly and, in fact, with any fractional power of choice, no matter how close to 0. In this regard, it remains open to characterize more precisely the optimal degree growth. What is the slowest degree growth at which partial distillability can exist? What is the slowest degree growth at which every graph sequence satisfying it is partially distillable?

Partial distillability also seems to have an intriguing connection with the notion of asymptotically robust GME in networks \cite{AGME1,AGME2}. This latter property can arise even if the edge-connectivity is bounded and, therefore, it does not imply partial distillability. However, all graph sequences that can be shown to be partially distillable by the results obtained here happen to have asymptotically robust GME (by Theorem 6.10 and Theorem 7.2 in \cite{AGME2}). Thus, it would be interesting to investigate if this property is a necessary condition for partial distillability. This would for instance imply that partial distillability could not exist at sublogarithmic degree growth (cf.\ Theorem 6.1 in \cite{AGME2}). As we have already mentioned, the problem of asymptotically robust GME can be related to partial distillability between arbitrary pairs of parties. An interesting open question in this direction is whether there exist graph sequences with partial distillability for a particular choice of the cardinality $m$ of the subsets of parties but that are not partially distillable for $m+1$.

Finally, different variations and generalizations of this problem can also be studied in future work. For instance, one could allow the cardinality $m$ of the subsets of parties to depend on the growing number of parties $N$. Our proof could in principle work as well to generalize to this case Theorem \ref{th:sufficient} if one restricts $m$ to be $o(\lambda(G_n))$ (cf.\ Eq.\ (\ref{mnthsuf})). However, one would need to study more carefully the speed at which $p'$ in Eq.\ (\ref{distillthsuf}) goes to 1 so as to compensate the increasing amount of noise that $\mathcal{T}^{\otimes(|V_0|-1)}$ in Eq.\ (\ref{teleportthsuf}) now introduces. The same comment applies to the result of Theorem \ref{th:slow} if $m$ is $o(|G_n|^\alpha)$ (cf.\ Eq.\ (\ref{numberspidergraphsthslow})). 

\begin{acknowledgments}
We thank Carlos Palazuelos for useful discussions. We acknowledge financial support from Comunidad de Madrid (grant QUITEMAD-CM TEC-2024/COM-84). A.B. acknowledges financial support from the Spanish Ministerio de Ciencia, Innovaci\'on y Universidades (grant PID2020-117477GB-I00 and through the UC3M Margarita Salas 2021-2023 program). J. I. de V. acknowledges financial support from the Spanish Ministerio de Ciencia, Innovaci\'on y Universidades (grant PID2023-146758NB-I00, grant PID2024-160539NB-I00 and ``Severo Ochoa Programme for Centres of Excellence" grant CEX2023-001347-S funded by MCIN/AEI/10.13039/501100011033).
\end{acknowledgments}

\end{document}